\documentclass[journal]{IEEEtran}
\usepackage{graphics,color}
\usepackage{graphicx}
\usepackage{amssymb,amsmath,amsfonts}
\usepackage{amsmath,mathrsfs,bm,url,times}
\usepackage{latexsym}
\usepackage{subfigure}

\DeclareMathOperator{\diag}{diag}
\renewcommand{\Re}{\mathrm{Re}}
\renewcommand{\Im}{\mathrm{Im}}

\newtheorem{proposition}{Proposition}

\newtheorem{lemma}{Lemma}

\newtheorem{corollary}{Corollary}
\newtheorem{remark}{Remark}

\newcommand{\R}{{\mathbb R}}
\newcommand{\C}{{\mathbb C}}

\newcommand{\oneton}{1,\dots,n}

\newcommand{\jj}{\mathbf j}

\newcommand{\ignore}[1]{}
\begin{document}
\title{Local pinning of networks of multi-agent systems with transmission and pinning delays}
\author{Wenlian Lu~\IEEEmembership{Senior Member,~IEEE} \and Fatihcan M. Atay\thanks{W. L. Lu is with the Centre for Computational Systems Biology and School of Mathematical Sciences, Fudan University, Shanghai, P. R. China (email: w.l.lu@ieee.org); F. M. Atay is with the Max Planck Institute for Mathematics in the Sciences, Leipzig, Germany (email: fatay@mis.mpg.de).}\thanks{The authors thank the ZiF (Center for Interdisciplinary Research) of Bielefeld University, where part of this research was conducted under the program \emph{Discrete and Continuous Models in the Theory of Networks}. FMA acknowledges the support of the European Union
(FP7/2007-2013)
under grant  \#318723 (MatheMACS). WLL acknowledges the support of the National Natural Sciences
Foundation of China under Grant (61273309)
and the Program for New Century Excellent Talents
in University (NCET-13-0139).}}

\maketitle
\begin{abstract}
We study the stability of networks of multi-agent systems with local pinning strategies and two types of time delays, namely the transmission delay in the network and the pinning delay of the controllers. Sufficient conditions for stability are derived  under specific scenarios by computing or estimating the dominant eigenvalue of the characteristic equation.
In addition, controlling the network by pinning a single node is studied. Moreover, perturbation methods are employed to derive conditions in the limit of small and large pinning strengths.
Numerical algorithms are proposed to verify  stability, and simulation examples are presented to confirm the efficiency of analytic results.
\end{abstract}

\section{Introduction}
Control problems in multi-agent systems have been attracting attention in diverse contexts \cite{DeG}--\cite{Li04}. In the consensus problem, for example, the objective is to make all agents converge to some common state by designing proper algorithms \cite{Reynolds87}-\cite{Fax},
such as the linear consensus protocol
\begin{equation}
\dot{x}_{i}=-\sum_{j=1}^{n}L_{ij}x_{j}(t),~i=\oneton. \label{ma}
\end{equation}
Here, $x_{i}\in \mathbb{R}$ is the state of agent $i$ and $L_{ij}$ are the components of the Laplacian matrix $L$,
satisfying $L_{ij}\le 0$ for all $i\ne j$ and $L_{ii}=-\sum_{j\ne i}L_{ij}$.
The Laplacian is associated with the underlying graph $\mathcal{G}$, whose links can be directed and weighted. It can be shown that, if the underlying graph has a spanning tree, then all agents converge to a common number, which depends on the initial values  \cite{DeG,Jadbabaie,Fax}.
On the other hand, if it is desired to steer the system to a prescribed consensus value, auxiliary control strategies are necessary. Among these, \emph{pinning control} is particularly attractive  because it is easily realizable by controlling only a few
agents, driving them to the desired value $s$ through feedback action:
\begin{equation}
\dot{x}_{i}=-\sum_{j=1}^{n}L_{ij}x_{j}(t)-\delta_{\mathcal D}(i)c(x_{i}-s),\quad i=\oneton, \label{pin}
\end{equation}
where $\mathcal D$ denotes the subset of agents  where feedback is applied, with cardinality $|\mathcal D|=m$, $\delta_{\mathcal D}(i)$ is the indicator function (1 if $i\in \mathcal{D}$ and 0 otherwise), and $c>0$ is the pinning strength. Eq. (\ref{pin}) provides the local strategy that pins a few nodes to stabilize the whole network at a common desired value. The following hypothesis is natural in pinning problems and assumed in this paper.

\textbf{(H)} Each strongly connected component of $\mathcal G$ without incoming links from the outside has at least one node in $\mathcal D$.

The following result is proved in \cite{CLL,LLR}.

\begin{proposition}\label{thm0}
If (H) holds, then system (\ref{pin}) is asymptotically stable at $x_{i}=s$  $\forall i$.
\end{proposition}

In many networked systems, however, time delays inevitably occur due to limited information transmission speed; so Proposition~\ref{thm0} does not apply.
 In this paper we consider systems with both transmission and pinning delays,
\begin{equation}
\dot{x}_{i}=-\sum_{j=1,j\ne i}^{n}L_{ij}(x_{j}(t-\tau_{r})-x_{i}(t))-c\delta_{\mathcal D}(i)(x_{i}(t-\tau_{p})-s), \label{pintransdelay}
\end{equation}
for $i=\oneton$, where $\tau_{r}$ denotes the {\em transmission delay} in the network and $\tau_{p}$ is the {\em pinning delay} of the controllers.
Several recent papers have addressed the stability of consensus systems with various delays. It has been shown that consensus can be achieved under transmission delays if the graph has a spanning tree \cite{Olf04}-\cite{Fatihcan2}.
However, if a sufficiently large delay is present also in the self-feedback of the node's own state, then consensus may be destroyed \cite{Bli}; similar conclusions also hold in cases of time-varying topologies \cite{More}--\cite{Lu1} and heterogeneous delays \cite{Xiang}-\cite{And}.
The stability of pinning networks with nonlinear node dynamics have been studied in \cite{Wang02}--\cite{Song4}, \cite{LiuChen}--\cite{Song3}.
However, the role of pinning delay was considered in only a few papers \cite{LiuChen}--\cite{Song3}, where it was argued that stability can be guaranteed if the pinning delays are sufficiently small. Precise conditions on the pinning delay for stability, the relation to the network  topology, and the selection of pinned nodes have not yet been addressed.

In this paper, we study the stability of the model (\ref{pintransdelay}) under both transmission and pinning delays. 
First, we derive an estimate of the largest admissible pinning delay.
Next, we consider several specific scenarios and present numerical algorithms to verify stability by calculating the dominant eigenvalue of the system.
Included among the scenarios are the cases
when only a single node is pinned in the absence of transmission delay, or when the transmission and pinning delays are identical.
Finally, we use a perturbation approach to estimate
the dominant eigenvalue for very small and very large pinning strengths.

\section{Notation and Preliminaries}

A directed graph $\mathcal{G}=\{\mathcal V, \mathcal E\}$ consists of a node set $\mathcal{V}=\{v_{1},\dots,v_{n}\}$ and a  link set $\mathcal{E}\subseteq \mathcal{V}\times \mathcal{V}$.
A (directed) {\em path} of length $l$ from node $v_{j}$ to $v_{i}$, denoted $(v_{r_{1}},\dots,v_{r_{l+1}})$, is a sequence of $l+1$ distinct vertices with $v_{r_{1}}=v_{i}$ and $v_{r_{l+1}}=v_{j}$ such that $(v_{r_{k}},v_{r_{k+1}})\in \mathcal{E}$ for $k=1,\dots,l$.
The graph is called strongly connected if there is a directed path from any node to any other node, and it is said to have a spanning tree if there is a node $v_{p}\in\mathcal V$ such that for any other node $j$ there is a path
from $v_{p}$ to $v_j$.

We denote the imaginary unit by $\mathbf{j}$ and the $n\times n$ identity matrix by $I_n$.
For a matrix $L$, $L_{ij}$ denotes its $(i,j)^{th}$ element and $L^{\top}$ its transpose.
The Laplacian matrix $L$ is associated with the graph $\mathcal G$ in the sense that there is a link from $v_{j}$ to $v_{i}$ in $\mathcal{G}$ if and only if  $L_{ij}\neq 0$.
We denote the eigenvalues of $L$ by $\{\theta_1,\dots,\theta_n\}$.
Recall that zero is always an eigenvalue, with the corresponding eigenvector $[1,\dots,1]^\top$,
and $\Re(\theta_i)> 0$ for all nonzero eigenvalues  $\theta_ i$.
Furthermore, if the graph $\mathcal{G}$ is strongly connected (or equivalently, if $L$ is irreducible), then zero is a simple eigenvalue of $L$.
The diagonal element $L_{ii}$ is the {\em weighted in-degree} of node $i$.
Let $K=\diag\{L_{11},\dots,L_{nn}\}$ be the diagonal matrix of in-degrees and $A=K-L$.
Let $y_{i}=x_{i}-s$, $y=[y_{1},\dots,y_{n}]^{\top}$, and $D=\diag\{d_{1},\cdots,d_{n}\}$ with $d_{i}=\delta_{\mathcal{D}}(i)$.
System (\ref{pintransdelay}) can be rewritten as
\begin{equation}
\dot{y}=-Ky+Ay(t-\tau_{r})-cDy(t-\tau_{p}).\label{pintransdelayM}
\end{equation}
Considering solutions in the form $y(t)=exp(\lambda t)\xi$ with $\lambda\in\C$ and $\xi\in\C^{n}$, the characteristic equation of (\ref{pintransdelayM}) is obtained as
\begin{equation}
\chi(\lambda) :=  \det\left[\lambda I_{n}+K-A\exp(-\lambda\tau_{r})+cD\exp(-\lambda\tau_{p})\right]=0.\label{chpintransdelay}
\end{equation}The asymptotic stability of (\ref{pintransdelayM}) is equivalent to all characteristic roots $\lambda$ of (\ref{chpintransdelay}) having negative real parts.
The root having the largest real part will be termed as the dominant root or the dominant eigenvalue.
For the undelayed case, Proposition~\ref{thm0} can be equivalently stated as follows.

\begin{corollary} \label{cor0}
If (H) holds, then all eigenvalues of $L+cD$ have negative real parts.
\end{corollary}

We also state an easy observation for later use:

\begin{lemma}\label{lem0}
  For any two column vectors $u,v\in\R^{n}$,
 $\det(I_{n}+uv^{\top})=1+v^{\top}u.$
\end{lemma}

\section{Estimation of the largest admissible pinning delay}


We first show that the system (\ref{pintransdelayM}) is stable for all values of the pinning delay $\tau_p$ smaller than a certain value $\tau_p^*$.

\begin{proposition}\label{thm1}
Assume condition {(H)}. Let
\begin{equation} \label{F}
F(w,c,l,\tau)= c^{2}+\omega^{2}+2c\left[l\cos(\omega\tau)
-\omega\sin(\omega\tau)\right]
\end{equation}
 and define
\begin{equation}
\tau_{p}^{*}=\sup_{\tau>0}\left\{\tau:~\min_{\omega\in\R}\min_{i\in\mathcal D}F(w,c,L_{ii},\tau)>0\right\}.  \label{taup}
\end{equation}
If $\tau_{p}<\tau_{p}^{*}$, then system (\ref{pintransdelayM}) is stable for all $\tau_{r}\ge 0$.
\end{proposition}

\begin{proof}
First, we take $\tau_{p}=0$ and prove stability for all $\tau_r\ge0$. Assume for contradiction  that there exists some characteristic root $\lambda^{*}$ of (\ref{chpintransdelay}) such that $\Re(\lambda^*) \ge 0$. Applying the Gershgorin disc theorem to \eqref{chpintransdelay}, we have
\begin{equation}
|\lambda^{*}+L_{ii}+cd_{i}|\le\sum_{j\ne i}|L_{ij}||\exp(-\lambda^{*}\tau_{r})|\le \sum_{j\ne i}|L_{ij}|=L_{ii} \label{disc0}
\end{equation}
for some $i$, which implies
\begin{equation*}
[\Re(\lambda^{*})+L_{ii}+cd_{i}]^{2}+[{\Im}(\lambda^{*})]^{2}\le L_{ii}^{2}. \label{disc1}
\end{equation*}
Since $L_{ii},c,d_{i}\ge 0$, it must be the case that $\Re(\lambda^{*})=\Im(\lambda^*)=0$; i.e., $\lambda^* = 0$.
%
Then  $\exp(-\tau_{r}\lambda^{*})=1$, and since $\tau_p=0$,  \eqref{chpintransdelay} gives $\det(\lambda^* I_{n}+L+cD)=0$.
This, however, contradicts Corollary~\ref{cor0}. Therefore, when $\tau_{p}=0$, all characteristic roots of (\ref{chpintransdelay}) have negative real parts.

We now let $\tau_p \ge 0$. Suppose (\ref{chpintransdelay}) has a purely imaginary  root $\lambda=\jj\omega$, $\omega\in\mathbb{R}$. By (\ref{disc0}), we have, for some index $q$,
\begin{eqnarray*}
&&|\jj\omega+L_{qq}+cd_{q}\exp(-\jj\omega\tau_{p})|\le\sum_{j\ne q}|L_{qj}||\exp(-\jj\omega\tau_{r})\\
&&=\sum_{j\ne q}|L_{qj}|=L_{qq}
\end{eqnarray*}
implying
\begin{equation*}
\sqrt{[L_{qq}+cd_{q}\cos(\omega\tau_{p})]^{2}+[\omega-cd_{q}
\sin(\omega\tau_{p})]^{2}}\le L_{qq}.
\end{equation*}
Thus,
\begin{equation}
(cd_{q})^{2}+\omega^{2}+2cd_{q}\left(L_{qq}
\cos(\omega\tau_{p})-\omega\sin(\omega\tau_{p})\right)\le 0.
\label{ineq1}
\end{equation}
We claim that $q$ must be a pinned node. For if $d_{q}=0$, then $\omega$ must be zero, which implies that zero is a characteristic root of (\ref{chpintransdelay}), contradicting Corollary~\ref{cor0}. Therefore $d_{q}=1$.
In the notation of \eqref{F}, the inequality \eqref{ineq1} can then be written as $F(w,c,L_{qq},\tau_{p})\le 0$.
By (\ref{taup}), however, we have that $F(w,c,L_{qq},\tau_{p})>0$ for all $p\in\mathcal D$, $\omega\in\R$ and $\tau_{p}<\tau_{p}^{*}$.
We conclude that (\ref{chpintransdelay}) does not have purely imaginary roots for $\tau_{p}<\tau_{p}^{*}$.
Thus, by  \cite[Theorem 2.1]{RW},  all characteristic roots of (\ref{chpintransdelay}) have  strictly negative real parts for $\tau_{p}<\tau_{p}^{*}$.
\end{proof}
\begin{remark}
Proposition \ref{thm1} provides an estimate for the largest admissible pinning delay for which system (\ref{pintransdelayM}) is stable.
This estimate needs only the knowledge of the set of pinned nodes and their weighted in-degrees.
\end{remark}
\ignore{
and can be numerically computed in two steps.
First, for given $c,l,\tau$, we seek the global minimum value of $F(\omega,c,l,\tau)$ in \eqref{F} with respect to $\omega\in\R$ and denote it $G(c,l,\tau)=\min_{\omega\in\R}F(\omega,c,l,\tau)$.
Second, we solve $G(c,l,\tau)=0$ with respect to $\tau$ for the smallest positive solution, and denote it $H(c,l)$. Then, $\tau_{p}^{*}=\min_{i\in\mathcal D}H(c,L_{ii})$. In this case, the estimate is independent of $\tau_{r}$. Figure~\ref{Fig1} shows that $H(c,l)$ decreases with respect to both $c$ and $l$. That is, the larger the pinning strength and the degree of the pinned node, the smaller is $\tau_{p}^{*}$. Equivalently, a larger $\tau_{p}^{*}$ requires a smaller pinning strength or pinning lower-degree nodes to ensure stability. From the proof, one can see that Proposition 1 still holds for the case of heterogeneous transmission delays and a single pinning delay.

\begin{figure}[!t]
\centering
\subfigure[$H(c,l)$ versus $c$ for $l=5$]
{
\includegraphics[height=0.2\textwidth,width=.4\textwidth]
{tau_c_l5.eps}
}
\subfigure[$H(c,l)$ versus degree for $c=1$]{
\includegraphics[height=0.2\textwidth,width=.4\textwidth]
{tau_l_c1.eps}
}
\caption{Estimation of the largest admissible pinning delay as a function $H(c,l)$ of the pinning strength $c$ and the degree $l$.
}\label{Fig1}
\end{figure}}

\section{Pinning a single node}
We now consider the possibility of controlling the network using a single node, say, the $q${th} one. Then $D=u_{q}u_{q}^{\top}$, where $u_{q}$ denotes the $q$th standard basis vector,
whose $q$th component is one and other components zero.
If $\lambda I_{n}+K-A\exp(-\lambda\tau_{r})$ is nonsingular, the characteristic equation (\ref{chpintransdelay}) becomes
\begin{align}
\chi(\lambda) = & \det\left[\lambda I_{n}+K-A\exp(-\lambda\tau_{r})+cu_{q}u_{q}^{\top}\exp(-\lambda\tau_{p})\right] \nonumber\\
=&\det(\lambda I_{n}+K-A\exp(-\lambda\tau_{r}))\nonumber\\
&\det\left[I_{n}+cu_{q}u_{q}^{\top}(\lambda I_{n}+K-A\exp(-\lambda\tau_{r}))^{-1}\exp(-\lambda\tau_{p})\right]\nonumber\\
=&\det(\lambda I_{n}+K-A\exp(-\lambda\tau_{r}))\nonumber\\
&(1+c u_{q}^{\top}(\lambda I_{n}+K-A\exp(-\lambda\tau_{r}))^{-1}u_{q}\exp(-\lambda\tau_{p}))\label{eqnxx1}
\end{align}
using Lemma \ref{lem0}. Then we have the following result.

\begin{proposition}\label{prop3}
Assume (H). If all solutions $\lambda$ of the equation
\begin{equation}  \label{eqnx1}
1+c u_{q}^{\top}(\lambda I_{n}+K-A\exp(-\lambda\tau_{r}))^{-1}u_{q}\exp(-\lambda\tau_{p})=0
\end{equation}
satisfy $\Re(\lambda)<0$, then system (\ref{pintransdelayM}) is stable.
\end{proposition}

\begin{proof}
As in the first part of the proof of Proposition~\ref{thm1}, the equation
$\det[\lambda I_{n}+K-A\exp(-\lambda\tau_{r})]=0$
has no solutions with $\Re(\lambda)\ge 0$. Hence, if all solutions $\lambda$ of (\ref{eqnx1}) have negative real parts, then all  roots of (\ref{chpintransdelay}) have negative real parts.
\end{proof}

We consider two specific cases to obtain more information about the solutions of \eqref{eqnx1}. First, we consider the absence of transmission delays, i.e., $\tau_{r}=0$. Suppose for simplicity that $L$ is diagonalizable and has only real eigenvalues:  $L=Q^{-1}JQ$ for some nonsingular  $Q$ and a real diagonal matrix
$J=\diag\{\theta_{1},\dots,\theta_{n}\}$ of eigenvalues of $L$. The column vectors of $Q^{-1}$ (resp,  the row vectors of $Q$) are the right (resp., left) eigenvectors of $L$.
Then, (\ref{eqnx1}) can be written as
\begin{equation} \label{eqnx11}
1+c\zeta^{\top}(\lambda I_{n}+J)^{-1}\xi\exp(-\lambda\tau_{p})=0,
\end{equation}
where $\zeta^{\top}=u_{q}^{\top}Q$ is the $q$th left eigenvector and $\xi=Q^{-1}u_{q}$ is the $q$th right eigenvector of $L$. We expand \eqref{eqnx11}  as
\begin{equation}
1+c\sum_{i=1}^{n}\frac{\xi_{i}\zeta_{i}\exp(-\lambda\tau_{p})}
{\lambda+\theta_{i}}=0  \label{chsing}
\end{equation}in terms of the components $\xi_i,\zeta_i$ of $\xi$ and $\zeta$, respectively.
Consider the smallest value of $\tau_{p}$ for which there exists a purely imaginary solution, $\lambda=\jj\omega$. Then, the real and imaginary parts of (\ref{chsing}) give
\begin{equation*}
\begin{cases}
1+a(\omega)\cos(\omega\tau_{p})-b(\omega)\sin(\omega\tau_{p})&=0\\
b(\omega)\cos(\omega\tau_{p})+a(\omega)\sin(\omega\tau_{p})&=0
\end{cases}
\end{equation*}
where
\begin{equation}
a(\omega)=c\sum_{i=1}^{n}\frac{\xi_{i}\zeta_{i}\theta_{i}}{\omega^{2}+\theta_{i}^{2}},\quad
b(\omega)=c\sum_{i}\frac{\xi_{i}\zeta_{i}\omega}{\omega^{2}+\theta_{i}^{2}}.  \label{AB}
\end{equation}
Rearranging gives $\cos(\omega\tau_{p})=-a(\omega)/(a^{2}(\omega)+b^{2}(\omega))$ and $\sin(\omega\tau_{p})=b(\omega)/(a^{2}(\omega)+b^{2}(\omega))$. This implies $a(\omega)^{2}+b^{2}(\omega)=1$ and
\begin{equation}
\cos(\omega\tau_{p})=-a(\omega),\quad \sin(\omega\tau_{p})=b(\omega).\label{TT}
\end{equation}
We then have the following result.

\begin{proposition}  \label{single_pin_thm1}
Suppose $\tau_{r}=0$, $L$ is diagonalizable, irreducible, {and all its eigenvalues are real}. Let the eigenvalues $\{\theta_i\}$ of $L$ be sorted so that $\theta_{q}=0$,
and let $\zeta=[\zeta_{1},\dots,\zeta_{n}]$, $\sum_{k=1}^{n}\zeta_{k}=1$, be the left eigenvector of $L$ corresponding to the zero eigenvalue. Let $\mathcal Z$ denote the set of  positive solutions of the equation
\begin{equation}
	a^{2}(\omega)+b^{2}(\omega)=1\label{eqnx2}
\end{equation}
with respect to the variable $\omega^{2}$,
where $a(\omega)$ and $b(\omega)$ are given by (\ref{AB}).
Define
\begin{equation}
\tau_{p}^{M}=\frac{\arccos(-a(\sqrt{\max\mathcal Z}))}{\sqrt{\max\mathcal Z}}.\label{tauPP}
\end{equation}
Then system (\ref{pintransdelayM}) is stable for $\tau_{p}<\tau^{M}_{p}$.
\end{proposition}

\begin{proof}
Eq.~(\ref{eqnxx1}) implies that any purely imaginary solution  $\jj\omega$ of  (\ref{chpintransdelay}) should also be a solution of (\ref{chsing}).
Then $\omega$ must be a real solution of (\ref{eqnx2}).
By the definition of $\mathcal Z$, the solution set of (\ref{eqnx2})
with respect to $\omega$ is $\{\pm\sqrt{z}:~z\in\mathcal Z\}$.
%
{By the assumption of irreducibility, $\theta_{i}>0$ for all $i\ne q$ and $\zeta_{i},\xi_{i}>0$ $\forall i$.}
If $\omega=\sqrt{z}$, then the smallest positive solution of (\ref{TT}) with respect to $\tau_{p}$ is $\arccos(-a(\sqrt{z}))/\sqrt{z}$.
If, on the other hand,
$\omega=-\sqrt{z}$, noting that $a(\omega)>0$ and $b(\omega)\le 0$, the smallest positive solution of (\ref{TT})
is again $\arccos(-a(\sqrt{z}))/\sqrt{z}$.
Therefore, given $\omega^{2}\in\mathcal Z$, the smallest nonnegative solution of (\ref{TT}) with respect to $\tau_{p}$ should be in the set $\{\arccos(-a(\sqrt{z}))/\sqrt{z}:z\in\mathcal Z\}$.
Since the mapping $z\mapsto \arccos(-a(\sqrt{z}))/\sqrt{z}$ is a decreasing function of $z>0$, the quantity $\tau_{p}^{M}$ defined in (\ref{tauPP}) is the smallest nonnegative solution of (\ref{TT}) with respect to $\tau_{p}$, given $\omega^{2}\in\mathcal Z$.
Hence, for $\tau_{p}<\tau_{p}^{M}$ (\ref{chsing}) does not have any purely imaginary solutions. Since for $\tau_{p}=0$ all characteristic roots of (\ref{chpintransdelay}) have negative real parts, we conclude
that all roots have negative real parts for $\tau_{p}<\tau_{p}^{M}$.
\end{proof}

\begin{remark}
By derivation, Eq.~(\ref{chsing}) is independent of the ordering of the eigenvalues or the eigenvectors in $J$.
Therefore, the bound $\tau_{p}^{M}$ for allowable pinning delays given in Proposition~\ref{single_pin_thm1} does not depend on the choice of the pinned node.
\end{remark}

Proposition \ref{single_pin_thm1} suggests
an 
algorithm to calculate $\tau_{p}^{M}$:
\begin{enumerate}
\item Find the largest positive solution $\omega^2$ of the equation
  \begin{equation}
    \sum_{k=1}^{n}\frac{(\xi_{k}\zeta_{k})^{2}}{\omega^{2}+\theta_{k}^{2}}+2\sum_{i>j}
    \frac{\xi_{i}\xi_{j}\zeta_{i}\zeta_{j}(\theta_{i}\theta_{j}+\omega^{2})}
    {(\omega^{2}+\theta_{i}^{2})(\omega^{2}+\theta_{j}^{2})}=\frac{1}{c^{2}}. \label{eqnx3}
  \end{equation}
\item Calculate (\ref{tauPP}).
\end{enumerate}

We illustrate this approach in an Erd\H{o}s-Renyi (E-R) random network of $n=100$ nodes with linking probability $0.03$, where the first node is pinned.
The left and right eigenvectors of $L$ associated with the zero eigenvalue are given by $\zeta=[1,\dots,1]/\sqrt{n}$.
Figure \ref{single_region} shows the parameter region $\{(c,\tau_{p}):\tau_{p}<\tau_{p}^{M}\}$,
illustrating the inverse dependence of $\tau_p^M$ on $c$. Note that  $\tau_{p}>\tau_{p}^{M}$ does not necessarily imply instability, since  Proposition \ref{single_pin_thm1} gives only a sufficient condition. Nevertheless, the curve shown in Fig.~\ref{single_region}(a) turns out to be a good approximation of the boundary of the  exact stability region.
To illustrate, we take two parameter points very close ($\pm 10\%$ of the $\tau_{p}^{M}$) to the curve but on different sides of it, as indicated by blue and red stars in Fig.~\ref{single_region}(a).  We simulate (\ref{pintransdelay}) at the corresponding parameter values, with the same Laplacian as above and $\tau_{r}=0$. As seen in Fig.~\ref{single_region}(b)--(c), the two points indeed yield different stability properties.

\begin{figure}[!t]
\centering
\subfigure[The stability region $\{(c,\tau_{p}):\tau_{p}<\tau_{p}^{M}\}$.]{
\includegraphics[height=0.15\textwidth, width=0.4\textwidth]{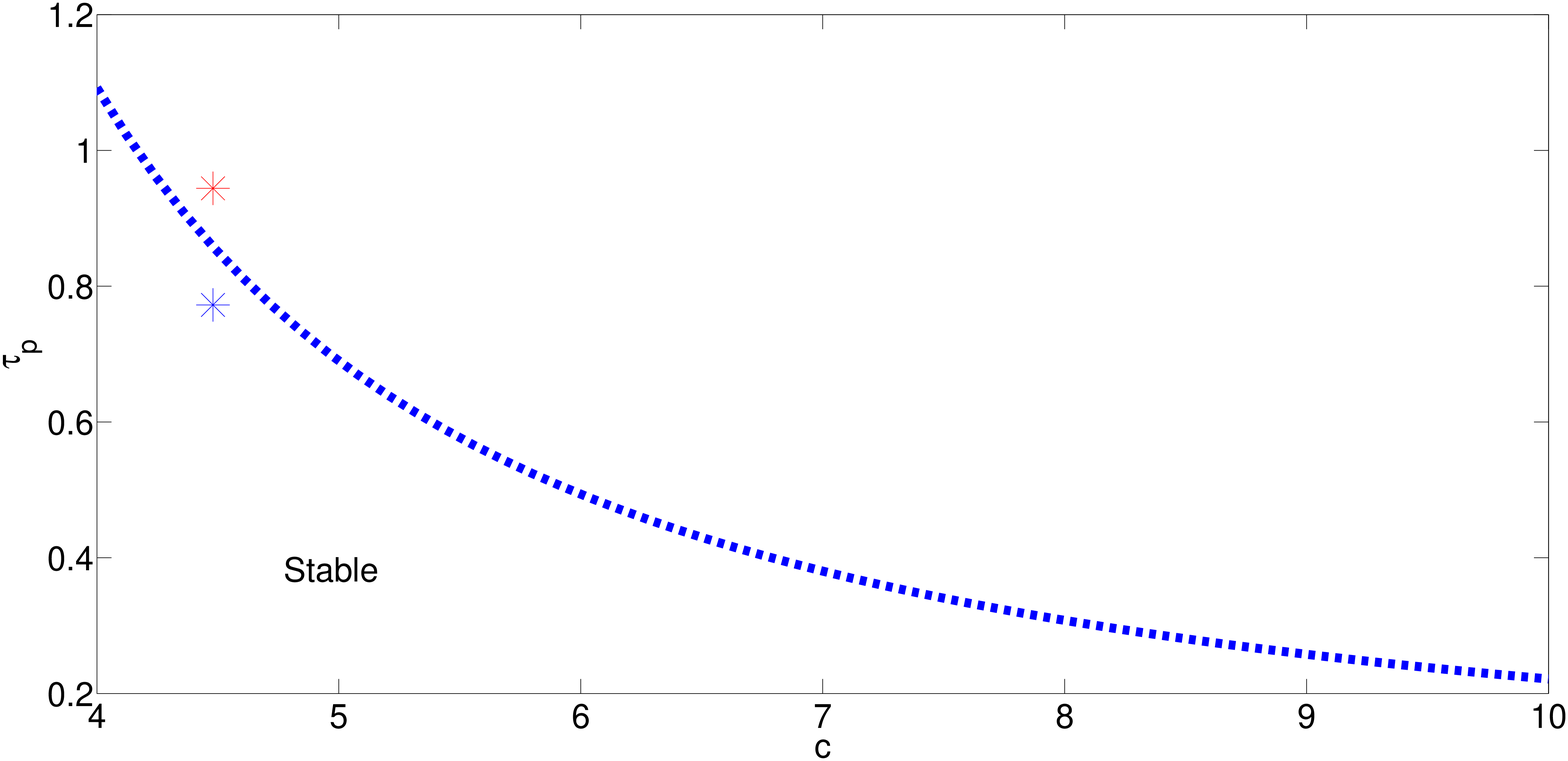}}
\subfigure[$c=4.48$ and $\tau_{p}=0.7724$]
{
\includegraphics[height=0.15\textwidth, width=.4\textwidth]
{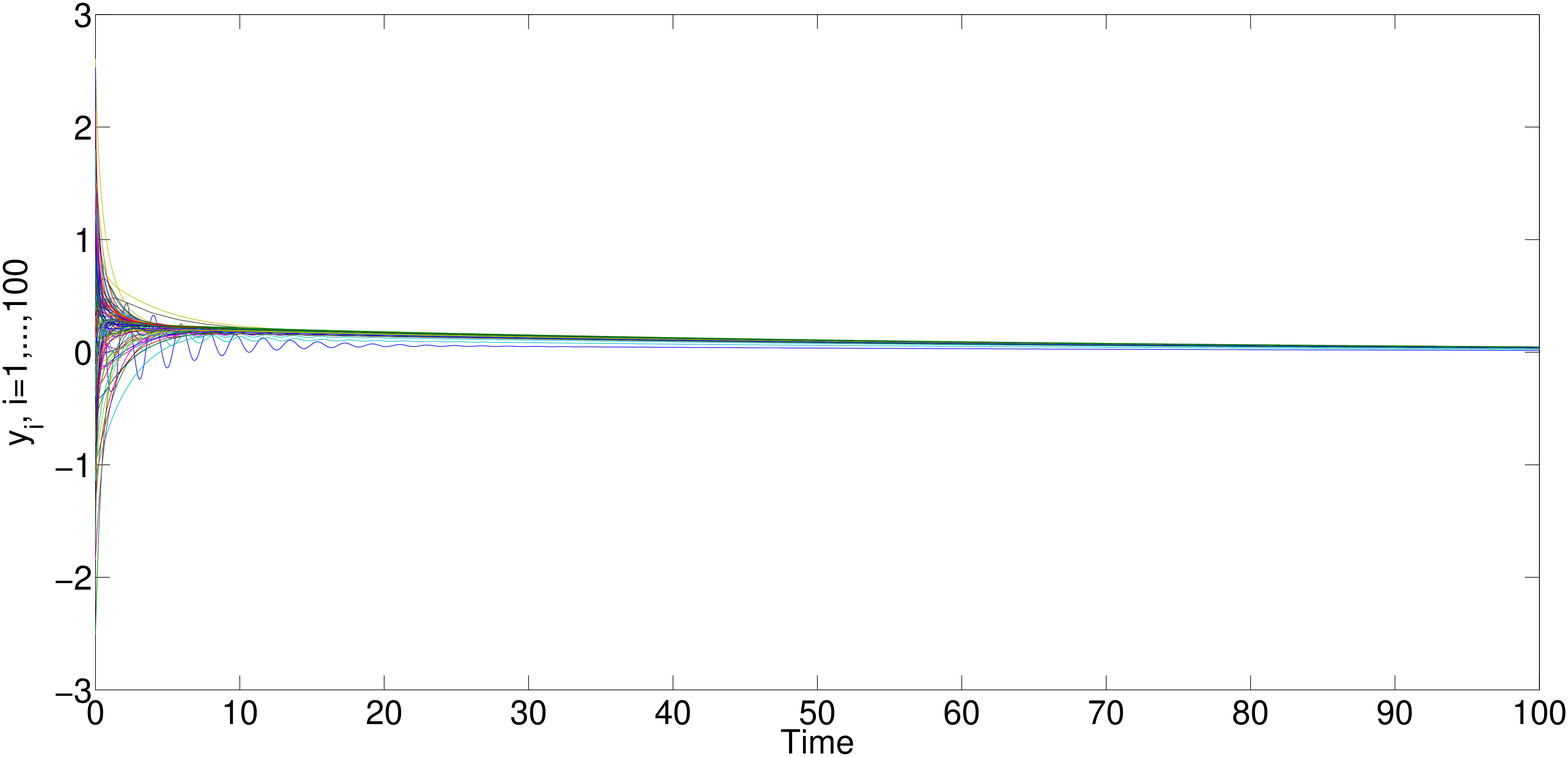}
}
\subfigure[$c=4.48$ and $\tau_{p}=0.9441$]{
\includegraphics[height=0.15\textwidth, width=.4\textwidth]
{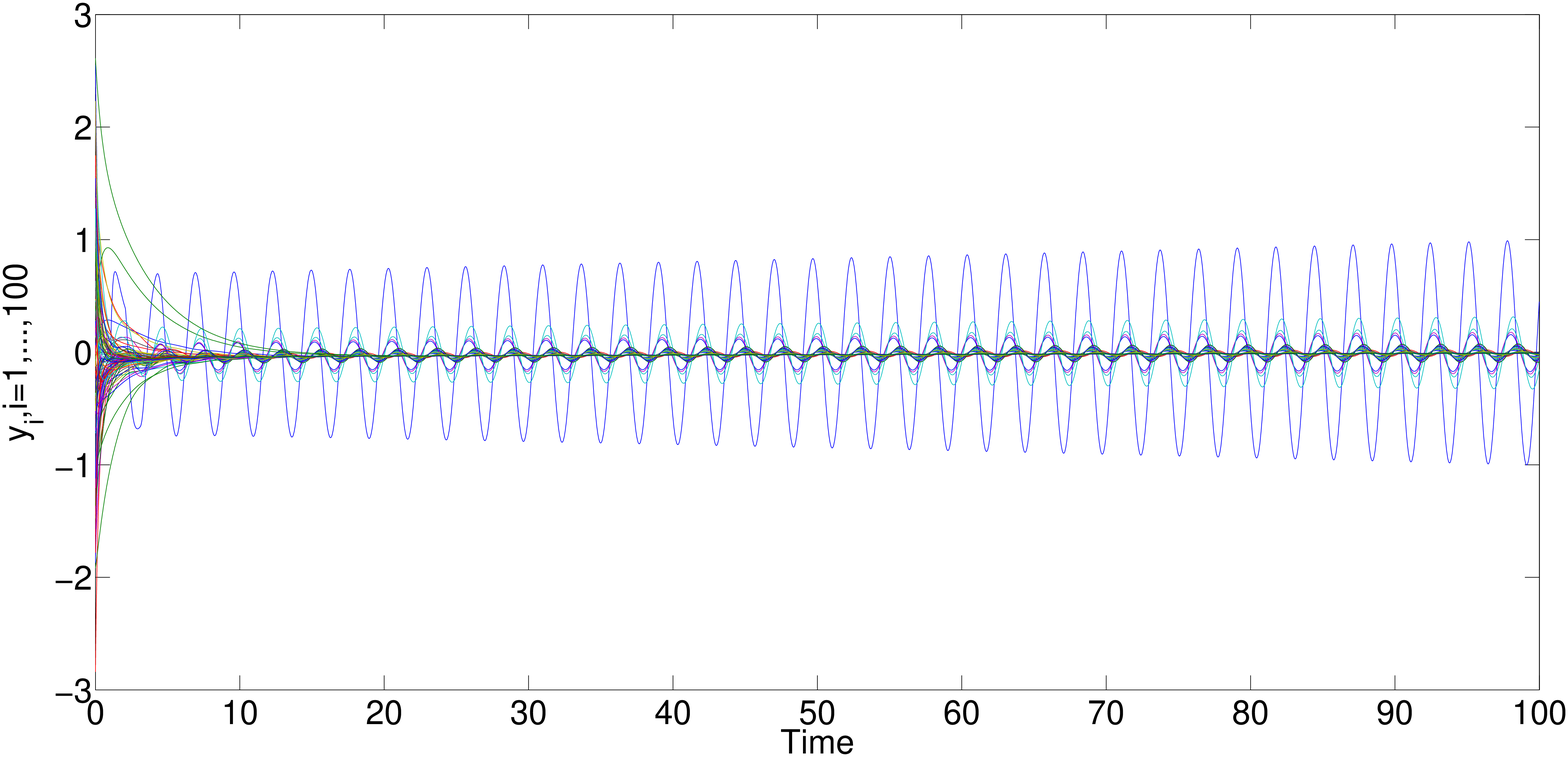}
}
\caption{(a) The stability region $\{(c,\tau_{p}):\tau_{p}<\tau_{p}^{M}\}$ in the parameter plane $(c,\tau_{p})$, where the dashed line depicts $\tau_p^M$ as a function of $c$. Direct simulation verifies that the system is indeed stable for the parameter values $c=4.48$ and $\tau_{p}=0.7724$ (b), and unstable for the slightly different values $c=4.48$ and $\tau_{p}=0.9441$ (c), corresponding to the blue and red stars, respectively, in subfigure (a).} \label{single_region}
\end{figure}

The other situation we consider is the homogeneous case when $L$ is diagonalisable and normalised, i.e., $L_{ii}=l$ $\forall i$
for some $l>0$, and
$\tau_{r}=\tau_{p}$. Then (\ref{eqnx1}) becomes
\begin{equation}
1+cu_{q}^{\top}((\lambda+l) I_{n}-A\exp(-\lambda\tau_{r}))^{-1}u_{q}\exp(-\lambda\tau_{p})=0.  \label{eqnx12}
\end{equation}
Let $L=QJQ^{-1}$;
thus $A=Q(l I_{n}-J)Q^{-1}$. Then, by the same algebra as above, (\ref{eqnx12}) becomes
\begin{equation}
1+c\sum_{k=1}^{n}\frac{\zeta_{k}\xi_{k}\exp(-\lambda\tau_{p})}{(\lambda+l)+(\theta_{k}-l)
\exp(-\lambda\tau_{p})}=0. \label{chsing1}
\end{equation}
We have the following result.

\begin{proposition}  \label{single_pin_thm2}
Suppose that $\tau_{r}=\tau_{p}$, $L$ is diagonalizable, irreducible, normalised ($L_{ii}=l$  $\forall i$), and all its eigenvalues $\{\theta_{i}\}$ are real.
Denote $\theta_{q}=0$ and let $\zeta=[\zeta_{1},\dots,\zeta_{n}]$ be the left eigenvector of $L$ corresponding to the eigenvalue $0$, with $\sum_{i}\zeta_{i}=1$.
Let $\mathcal S$ denote the set of all the branches of the solutions of
the equation
\begin{equation}
1+c\sum_{k=1}^{n}\frac{\zeta_{k}\xi_{k}}{\exp(- l\tau_{p})s/\tau_{p}+(\theta_{k}-l)}=0 \label{eqnx22}
\end{equation}
with respect to the variable $s$.
Then  system (\ref{pintransdelayM}) is stable whenever
the real parts of the numbers $\{\frac{W(s)}{\tau_{p}}-l:~s\in\mathcal S\}$ are all negative, where $W$ is the Lambert $W$ function \cite{Corless}.
\end{proposition}

Proposition \ref{single_pin_thm2} can be proved by transforming (\ref{chsing1}) into (\ref{eqnx22}) with $s=\tau_{p}(\lambda+l)\exp(\tau_{p}(\lambda+l))$ and using Proposition \ref{prop3}.



\section{Small and large pinning strengths} \label{small_pinning}

In this section, we consider the extreme situations when the pinning strength $c$ is very small or very large.
We will employ the perturbation approach in \cite{Tre,Li} to approximate the eigenvalues and eigenvectors in terms of $c$.

The characteristic roots $\lambda$ of (\ref{chpintransdelay}) are eigenvalues of the matrix $\Sigma(c,\lambda)=-K+A\exp(-\lambda\tau_{r})-cD\exp(-\lambda \tau_{p})$.
Hence, when $c=0$, the characteristic roots of (\ref{chpintransdelay}) equal to the eigenvalues $\{\sigma_{i}\}$ of $\Sigma(0,\lambda)$.
Under the condition (H), there is a single eigenvalue $\sigma_{1}=0$. We denote the right and left eigenvectors of $\Sigma(0,\sigma_{i})$ by $\phi^{i}$ and ${\psi^{i}}^{\top}$ respectively, with ${\psi^{i}}^{\top}\phi^{i}=1$.
It can be seen that $\psi^{1}$ and $\phi^{1}$ (associated with $\sigma_{1}=0$) are, respectively, the right and left eigenvectors of $L$ associated with the zero Laplacian eigenvalue.

Let $\lambda_{i}(c)$
denote the characteristic roots of (\ref{chpintransdelay}) and  $\tilde{\phi}^{i}(c)$ and $\tilde{\psi}^{i}(c)$ denote the right and left eigenvectors of $\Sigma(c,\lambda_{i}(c))$, regarded as functions of $c$, with $\lambda_{i}(0)=\sigma_{i}$, $\tilde{\phi}^{i}(0)=\phi^{i}$ and $\tilde{\psi}^{i}(0)=\psi^{i}$. Using a perturbation expansion \cite{Tre,Li},
\begin{eqnarray*}
\lambda_{i}(c)&=&\sigma_{i}+\lambda_{i}^{1}c+o(c),~
\tilde{\phi}^{i}(c)=\phi^{i}+\phi^{i,1}c+o(c)\\
\tilde{\psi}^{i}(c)&=&\psi^{i}+\psi^{i,1}c+o(c)
\end{eqnarray*}
where $o(c)$ denotes terms that satisfy $\lim_{c\to 0}|o(c)|/c=0$. Thus,
\begin{eqnarray*}
&&[-K+A\exp(-\lambda_{i}(c)\tau_{r})-cD\exp(-\lambda_{i}(c)\tau_{p})]\tilde{\phi}^{i}(c)\\
&&=\lambda_{i}(c)\tilde{\phi}^{i}(c).
\end{eqnarray*}
When $c$ is sufficiently small, the dominant eigenvalue is $\lambda_{1}(c)$, since $\sigma_{1}=0$ is the dominant eigenvalue  when $c=0$. Hence, we consider $i=1$.
Then
$\exp(-\lambda_{1}(c)\tau) = 1-c\lambda_{1}^{1}\tau+o(c)$.
Comparing the first-order terms in $c$ on both sides, $
(-A\lambda_{1}^{1}\tau_{r}-D)\phi^{i}-L\phi^{i,1}=\lambda_{1}^{1}\phi^{i}$.
Multiplying both sides with ${\psi^{1}}^{\top}$ and noting that ${\psi^{1}}^{\top}\phi^{1}=1$,
\begin{equation}
\lambda_{1}^{1}=-\frac{{\psi^{1}}^{\top}D\phi^{1}}{1+\tau_{r}
({\psi^{1}}^{\top}A\phi^{1})}.\label{lambda11}
\end{equation}
Hence, we have the following result.

\begin{proposition}\label{thm4} Suppose that the underlying graph is strongly connected and at least one node is pinned.
Then, for sufficiently small $c$, all characteristic roots of (\ref{chpintransdelay}) have negative real parts and the dominant root is given by
\begin{equation}  \label{approxc}
\lambda_{1}(c)=-\frac{{\psi^{1}}^{\top}D\phi^{1}}{1+\tau_{r}
({\psi^{1}}^{\top}K\phi^{1})}c+o(c).
\end{equation}
\end{proposition}

\begin{proof} Since the graph is strongly connected, $L$ has a simple zero eigenvalue. When $c=0$, the dominant root of (\ref{chpintransdelay}) is $\sigma_{1}=\lambda_{1}(0)$.
Since the roots of (\ref{chpintransdelay}) depend analytically on $c$, they are given by $\lambda_{1}(c)$ for all sufficiently small $c$.  Substituting (\ref{lambda11}) into $\lambda_{1}(c)$ and noting that ${\psi^{1}}^{\top}(-K+A)\phi^{1}=0$ completes the proof.
\end{proof}

In order to understand the meaning of (\ref{approxc}), consider the special case of an undirected graph with binary adjacency matrix $A$.
Then, with $\phi^{1}=[1,\dots,1]^{\top}$ and $\psi^{1}=[1,\dots,1]^{\top}/n$,
we have  ${\psi^{1}}^{\top}K\phi^1=\sum_{i=1}^{n}L_{ii}/n$, which equals the {\em average degree} of the graph.
In addition, ${\psi^{1}}^{\top}D\phi^{1}=\sum_{i=1}^{n}\delta_{\mathcal D}(i)/n$, which is the {\em fraction of pinned agents}.
Then, (\ref{approxc}) yields the approximation
\begin{equation}
\lambda_{1}(c)\approx-\frac{\textrm{Pinning~Fraction}}{1+\tau_{r}
\times\textrm{Mean~Degree}}c\label{approxc1}
\end{equation}
for small $c$, which uses only the pinning fraction and the mean degree of the graph. Since the real part of the dominant characteristic value measures the exponential convergence of the system,
Proposition \ref{thm4} implies that, for sufficiently small $c$, the convergence rate is improved if
the number of pinned nodes is increased, the transmission delay is reduced, or the mean degree is decreased. If the graph is directed,
a similar statement can be obtained by taking the components of $\psi^{1}$ as weights: $\psi^{1}D\phi^{1}=\sum_{j=1}^{n}\psi^{1}_{j}\delta_{\mathcal D}(j)$.

To illustrate this result, we employ a numerical method to calculate the real part of $\lambda_{1}(c)$, namely, by simulating the  system \eqref{pintransdelayM} and expressing its exponential convergence rate in terms of its largest Lyapunov exponent. In detail, letting $\tau_{m}=\max\{\tau_{r},\tau_{p}\}$, we partition time into disjoint intervals of length $\tau_{m}$, $t_{k}=k\tau_m$, and define $\eta_{k}(\theta)=y(t_{k}+\theta)$ for  $\theta\in[0,\tau_{m}]$. Then, the largest Lyapunov exponent, which equals to the largest real part of solutions of (\ref{chpintransdelay}), is numerically calculated via \cite{Farmer}
\begin{equation}
\Re(\lambda_{1,\mathrm{sim}})=\lim_{N\to\infty}\frac{1}{N\tau_{m}}\log\|\eta_{N}\|
=\lim_{N\to\infty}\frac{1}{N\tau_{m}}\sum_{k=1}^{N}\log
\frac{\|\eta_{k}\|}{\|\eta_{k-1}\|},  \label{LE}
\end{equation}
where $\|\cdot\|$ stands for the function norm. The latter is numerically calculated by approximating $\eta_{k}(\cdot)$ with a finite-dimensional vector  $\varphi_{k}$ obtained by evaluating $\eta_k$ at a finite number of equally spaced points and using the vector norm  $\|\varphi_{k}\|$.
The estimate \eqref{LE} can then be compared with the analytical estimate for $\Re(\lambda_{1})$ obtained from (\ref{approxc}):
\begin{equation}
\Re(\lambda_{1,\mathrm{est}})=-\frac{{\psi^{1}}^{\top}D\phi^{1}}{1+\tau_{r}
({\psi^{1}}^{\top}K\phi^{1})}c.
\label{est}
\end{equation}

For simulations, we generate an undirected E-R random graph of $n=100$ nodes with linking probability $p=0.03$
and randomly select a given fraction $f$ of them as the pinned nodes.
The pinning delay is taken as $\tau_{p}=0.1$.
Figure \ref{Fig4_small_c} shows that the simulated value of $\Re(\lambda_{1})$ decreases almost linearly with respect to $c$ and $f$, and increases with respect to $\tau_{r}$ and the mean degree. The simulation results are in a good agreement with the theoretical results. The error between $\Re(\lambda_{1,\mathrm{est}})$ and $\Re(\lambda_{1,\mathrm{sim}})$ depends on the values of $\lambda_{1}^{1}$ and $c$. It can be seen that the error will increase as $c$ or $\lambda_{1}^{1}$
(or equivalently, $f$) increases, or else as the mean degree or $\tau_{r}$ decreases.

\begin{figure}[htp]
\centering
\subfigure[Variation of $\Re(\lambda_{1})$ w.r.t. $c$][]
{
\includegraphics[width=.2\textwidth,height=.1\textwidth]
{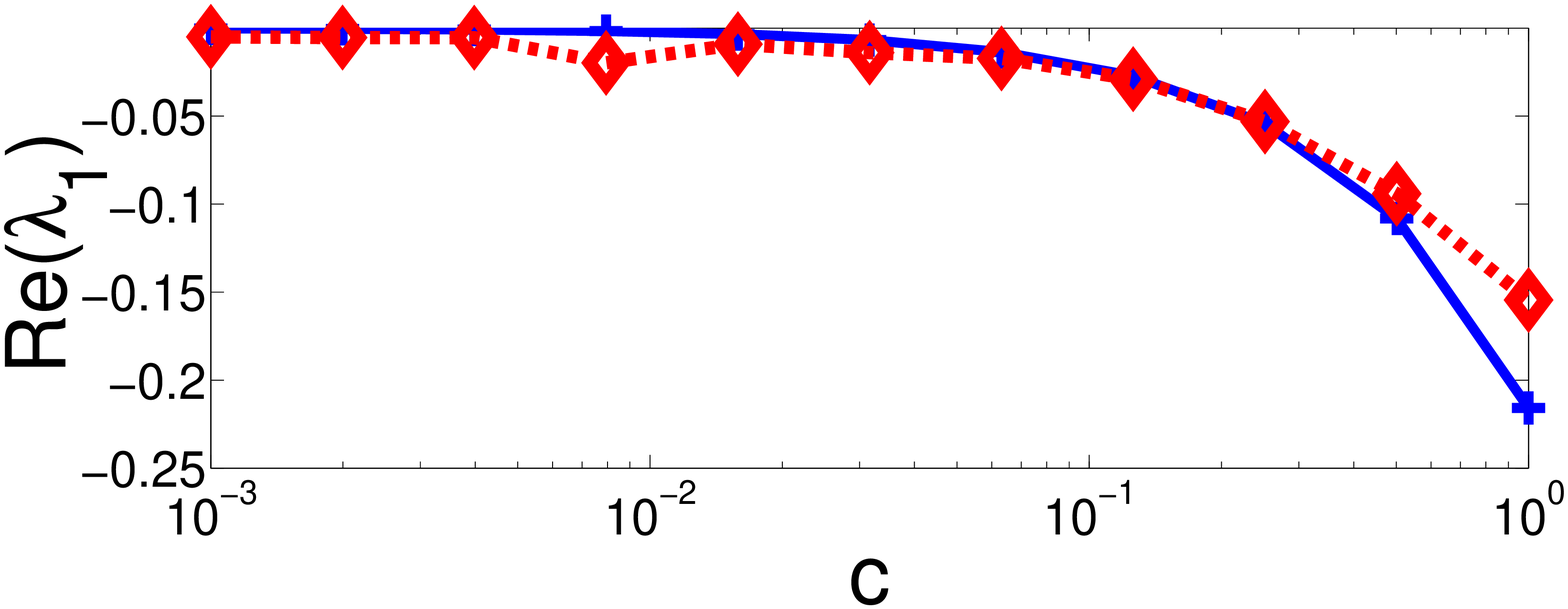}
}
\subfigure[Variation of $\Re(\lambda_{1})$ w.r.t. $f$][]{
\includegraphics[width=.2\textwidth,height=.1\textwidth]
{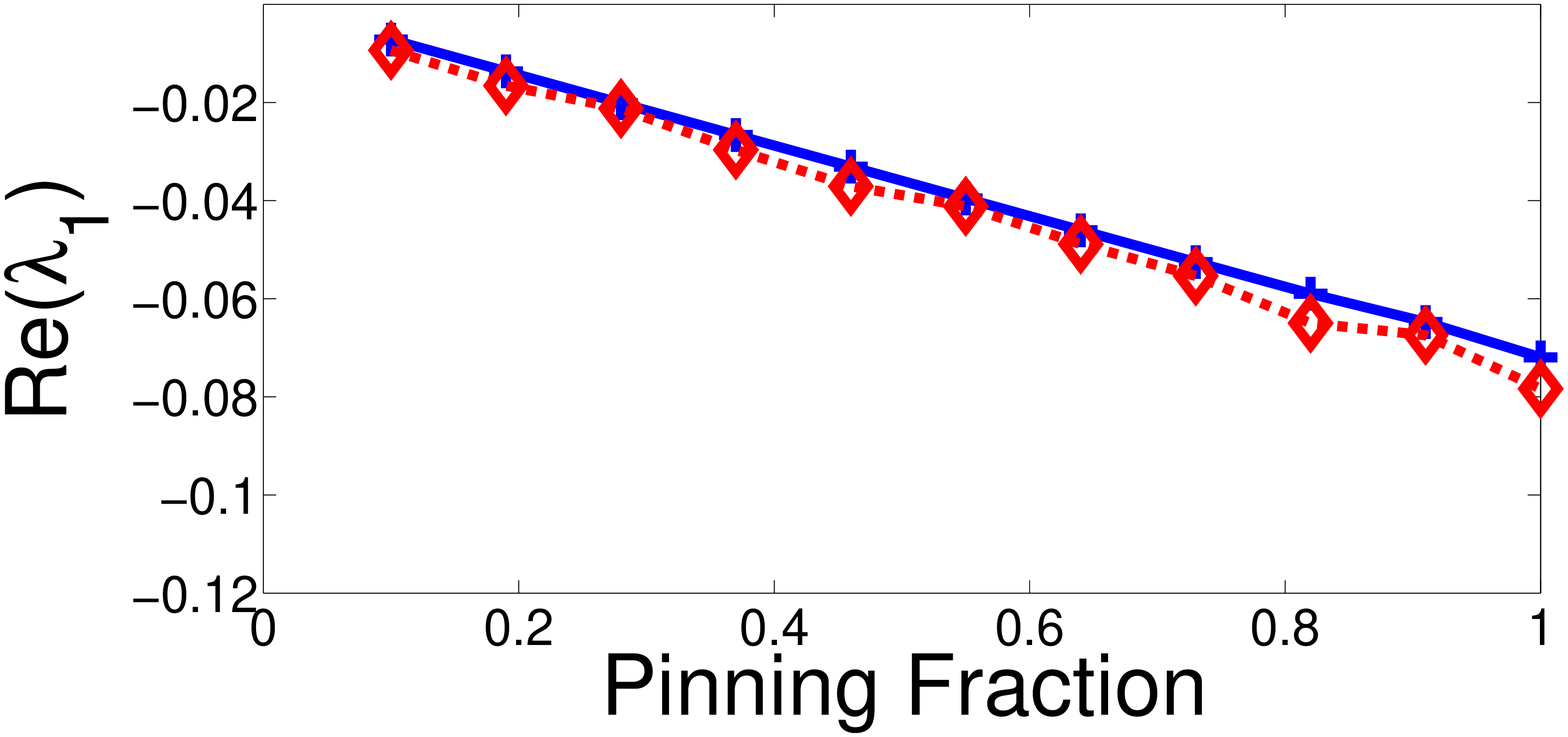}
}
\subfigure[Variation of ${\Re}(\lambda_{1})$ w.r.t. the mean degree][]{
\includegraphics[width=.2\textwidth,height=.1\textwidth]
{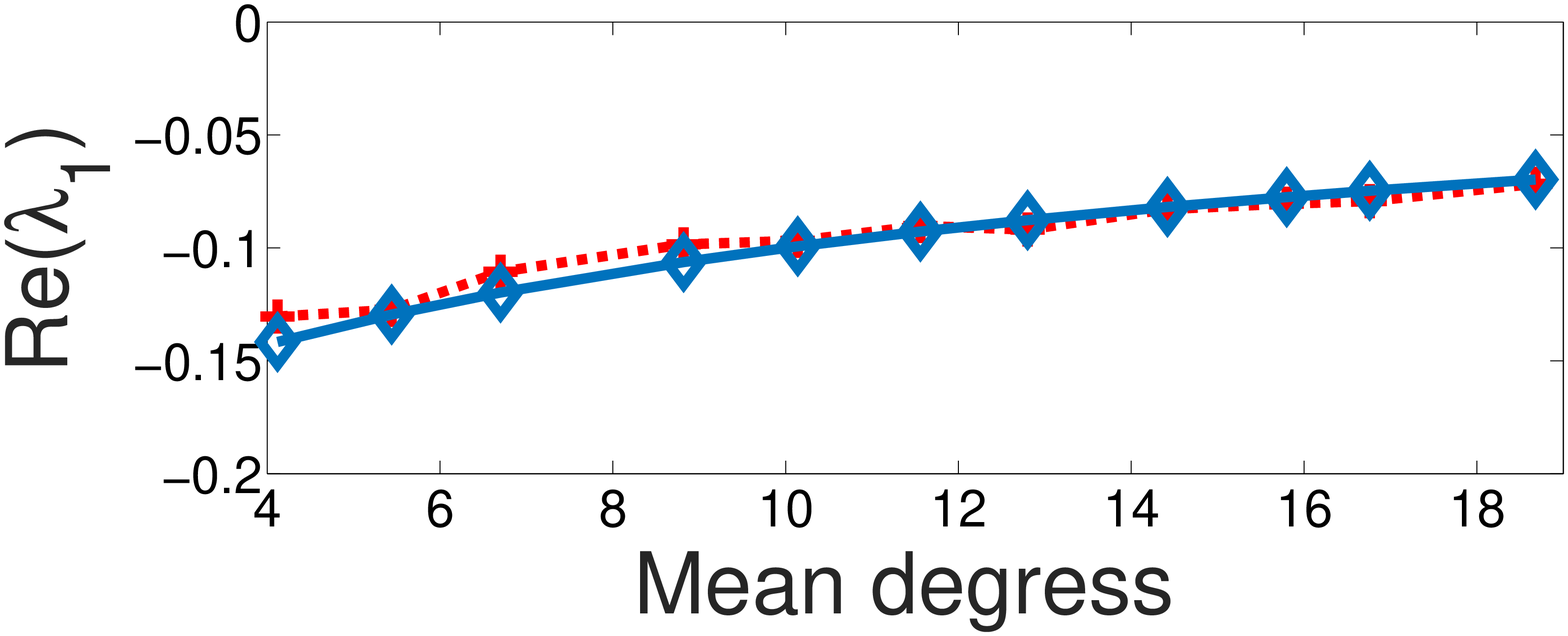}
}
\subfigure[Variation of ${\Re}(\lambda_{1})$ w.r.t. $\tau_{r}$][]{
\includegraphics[width=.2\textwidth,height=.1\textwidth]
{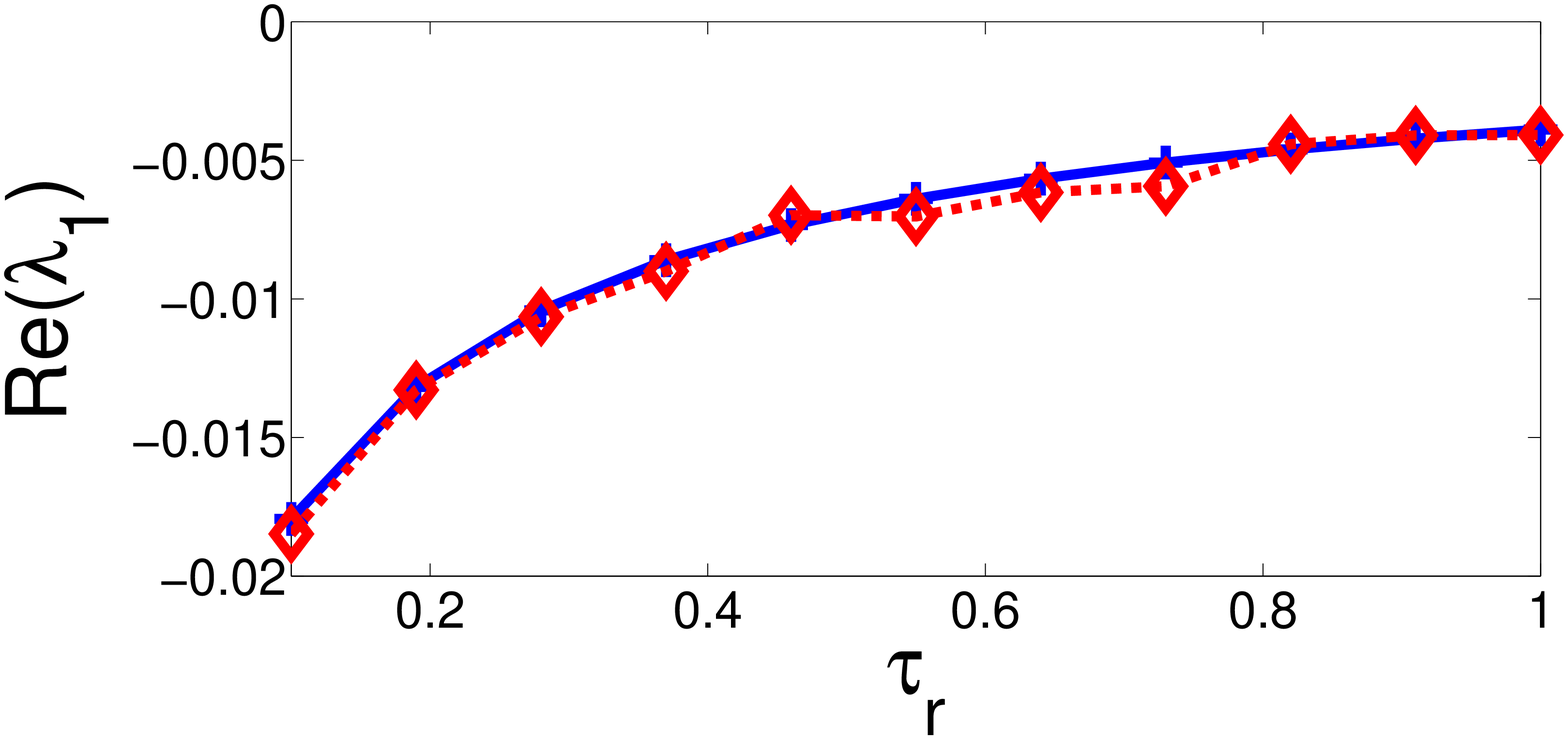}
}
\caption{Variation of ${\Re}(\lambda_{1})$ with
system parameters. The estimate \eqref{est} (plotted with {\color{blue} $+$}) shows good agreement with the values obtained via simulation and \eqref{LE} (plotted with {\color{red} $\diamond$}).
The parameters that are kept fixed are:
(a) $f=0.3$, $\tau_{r}=0.1$, mean degree $=3.4$;
(b) $c=0.1$, $\tau_{r}=0.1$, mean degree $=3.4$;
(c) $c=0.1$, $f=0.3$, $\tau_{r}=0.1$;
(d) $c=0.1$, $f=0.3$, mean degree $=3.4$.}
\label{Fig4_small_c}
\end{figure}

Next, we consider the case of large $c$. Letting $\epsilon=1/c$ and $\mu=\lambda/c$, (\ref{chpintransdelay}) is rewritten as
\begin{equation}
\det\left[\mu I_{n}+\epsilon K-\epsilon A\exp(-\mu \tau_{r}/\epsilon)
+D\exp(-\mu\tau_{p}/\epsilon)\right]=0.  \label{large_tau_ch}
\end{equation}
By the foregoing results, one can see that when $\epsilon$ is sufficiently small, equivalently, $c$ is sufficiently large, the largest admissible pinning delay for \eqref{pintransdelayM} approaches zero.
It is therefore natural to assume that $\tau_{p}$ depends on $c$ in such a way that $\tau_{p}c$ is bounded as $c$ grows large. Thus, we assume that $\tau_{pc}:=\tau_{p} c$ remains bounded as $c \to\infty$.

When $\epsilon=0$,
(\ref{large_tau_ch}) becomes approximately $
\dot{x}=-Dx(t-\tau_{p\infty})$,
where $\tau_{p\infty}$ can be any value between $\overline{\lim}_{c\to\infty}\tau_{pc}$ and  $\underline{\lim}_{c\to\infty}\tau_{pc}$. 
In terms of components, $
\dot{x}_{i}=
-x_{i}(t-\tau_{p\infty})$ if  $i\in\mathcal D$, and $0$ otherwise.
The characteristic equation
(\ref{large_tau_ch}) with $\epsilon=0$ can be written as
\begin{equation}
(\mu+\exp(-\mu\tau_{p\infty}))^{m}\mu^{n-m}=0\label{large_tau_ch_esp_0}
\end{equation}
where $m=|\mathcal{D}|$. It is known that $\Re(\mu)<0$ for all roots of the function $\mu \mapsto \mu+\exp(-\mu\tau_{p\infty})$ if and only if $\tau_{p\infty}<\frac{\pi}{2}$. Therefore, we impose the condition: $
\tau_{p} \, c<\frac{\pi}{2}$.

Thus, the largest real part of the solutions of (\ref{large_tau_ch_esp_0}) is zero, and is obtained for the solution $\mu=0$. The corresponding eigenspace
has dimension $n-m$ and has the form
\begin{equation*}
ES=\{u=[u_{1},\dots,u_{n}]^{\top}\in\mathbb R^{n}:~u_{i}=0,~\forall~i\in\mathcal D\}.
\end{equation*}
Without loss of generality, we assume $\mathcal D=\{1,\dots,m\}$.
Thus, we consider perturbation in terms of $\epsilon$ near zero eigenvalues $\mu_{i}$ and its corresponding right and left vectors, $\xi^{i},{\zeta^{i}}^{\top} \in ES$
such that $(\zeta^{i})^{\top}\xi^{i}=1$ and $(\zeta^{j})^{\top}\xi^{i}=0$ if $i\ne j$, $i,j=m+1,\dots,n$.
Let $\mu_{i}(\epsilon)$ stand for the perturbed solution of (\ref{large_tau_ch}), $\tilde{\xi}^{i}(\epsilon)$
and $\tilde{\zeta}^{i}(\epsilon)$ be the corresponding right and left eigenvectors, respectively. By a perturbation expansion,
\begin{eqnarray}
\mu_{i}(\epsilon)&=&\mu^{1}_{i}\epsilon+o(\epsilon),~
\tilde{\xi}^{i}(\epsilon)=\xi^{i}+\xi^{i,1}\epsilon+o(\epsilon)\nonumber\\
\tilde{\zeta}^{i}(\epsilon)&=&\zeta^{i}+\zeta^{i,1}\epsilon+o(\epsilon)\label{large_perturb}
\end{eqnarray}
as $\epsilon\to 0$. Thus, from (\ref{large_tau_ch}),
\begin{eqnarray*}
&&\left[-\epsilon K+\epsilon A\exp\left(-\mu_{i}(\epsilon)\frac{\tau_{r}}{\epsilon}\right)-D\exp(-\mu_{i}(\epsilon)\tau_{Pc})
\right]\tilde{\xi}^{i}(\epsilon)\\
&&=\mu_{i}(\epsilon)\tilde{\xi}^{i}(\epsilon).
\end{eqnarray*}
Since $\exp(-\mu_{i}(\epsilon)\tau)=1-\epsilon\mu_{i}^{1}\tau+o(\epsilon)$, by comparing the coefficients of order $1$, we have
\begin{eqnarray}
[-K+A\exp(-\mu^{1}_{i}\tau_{r})]\xi^{i}-D\xi^{i,1}=\mu^{1}_{i}\xi^{i}.\label{large_order_1}
\end{eqnarray}
We write
\begin{equation*}
K=\left[\begin{array}{ll}K_{1}&0\\
0&K_{2}\end{array}\right], \;
A=\left[\begin{array}{ll}A_{11}&A_{12}\\
A_{21}&A_{22}\end{array}\right], \;
D=\left[\begin{array}{ll}I_{m}&0\\
0&0\end{array}\right],
\end{equation*}
and $
\xi^{i}=[{\xi^{i}_{1}}^{\top}$, ${\xi^{i}_{2}}^{\top}]^{\top}$, $ \xi^{i,1}=[{\xi^{i,1}_{1}}^{\top},{\xi^{i,1}_{2}}^{\top}]^{\top}$,
with $K_{1}$, $A_{11}$, $\xi^{i}_{1}=0$ and $\xi^{i,1}_{1}$ corresponding to the pinned subset $\mathcal D$ of dimension $m$. Then (\ref{large_order_1}) becomes
\begin{equation}
\begin{cases}
[-K_{2}+A_{22}\exp(-\mu^{1}_{i}\tau_{r})]\xi^{i}_{2} &=\mu^{1}_{i}\xi^{i}_{2}\\
\exp(-\mu^{1}_{i}\tau_{r}) A_{12}\xi^{i}_{2}-\xi^{i,1}_{1} &= 0.
\end{cases}
\label{large_order_1x}
\end{equation}
We have the following result.

\begin{proposition}\label{thm5}
Suppose that the underlying graph is strongly connected and at least one node is pinned. Fix $\tau_{r}\ge 0$,  and suppose $\tau_{p} c<\frac{\pi}{2}$ as $c \to \infty$.  Then the dominant root of (\ref{large_tau_ch}) has the form
\begin{equation}
\lambda(c)=\mu^{1}_{*}+o(1) \; \text{ as } c\to\infty, \label{lambda_large_c}
\end{equation}
where $\mu^{1}_{*}$ is the dominant eigenvalue of the delay-differential equation
\begin{equation}
\dot{y}=-K_{2}y(t) + A_{22}y(t-\tau_{r}). \label{large_c_equ}
\end{equation}
Furthermore, $\Re (\lambda(c))<0$ for all sufficiently large $c$.
\end{proposition}

\begin{proof}
The condition $\tau_{pc}<\pi/2$ implies that, when $\epsilon=0$, the dominant root of the characteristic equation (\ref{large_tau_ch}) is zero and corresponds to the eigenspace $ES$. So, for sufficiently small $\epsilon$,
the dominant root of equation (\ref{large_tau_ch}) and the corresponding eigenvector have the form (\ref{large_perturb}), where $\mu_{i}^{1}$ satisfies the first equation in (\ref{large_order_1x}), i.e., is an eigenvalue of (\ref{large_c_equ}).
Since $\lambda(\epsilon)=\mu/\epsilon$, (\ref{lambda_large_c}) follows.
Moreover, since $-K_{2}+A_{22}$ is diagonally dominant, one can see that $\Re(\mu^{1}_{i})<0$ under condition (H). Therefore, for sufficiently large $c$, all characteristic values of system (\ref{pintransdelay}) have negative real parts.
\end{proof}

We note that $\mu_{*}^{1}$ depends only on the coupling structure of the uncoupled nodes. To illustrate this result, we consider examples with a similar setup as in Sec.~\ref{small_pinning}. We take an E-R graph with $n=100$ nodes and linking probability $p=0.03$, and pin $m=30$ nodes.
We set $\tau_{r}=0.1$ and $\tau_{p}=\frac{1}{c}$.
The real part of the dominant characteristic root of (\ref{chpintransdelay}) is numerically calculated via the largest Lyapunov exponent, using formula (\ref{LE}). Its theoretical estimation comes from Theorem \ref{thm5}: $
{\Re}(\lambda_{1,\mathrm{est}})=\max\left\{{\Re}(\mu^{1}) :\ \det\left(\mu^{1}I_{m}+K_{2}-A_{22}\exp(-\mu^{1}\tau_{r})\right)=0\right\}$,
where the largest real part of $\mu^{1}$ is similarly calculated from the largest Lyapunov exponent of (\ref{large_c_equ}).
Fig. \ref{Fig5} shows that as $c$ grows large, the real part of the dominant root of (\ref{chpintransdelay}) obtained from simulations approach the theoretical result ${\Re}(\lambda_{1,\mathrm{est}})$, thus verifying Proposition~\ref{thm5}.

\begin{figure}[!t]
\centering
{
\includegraphics[width=.4\textwidth]
{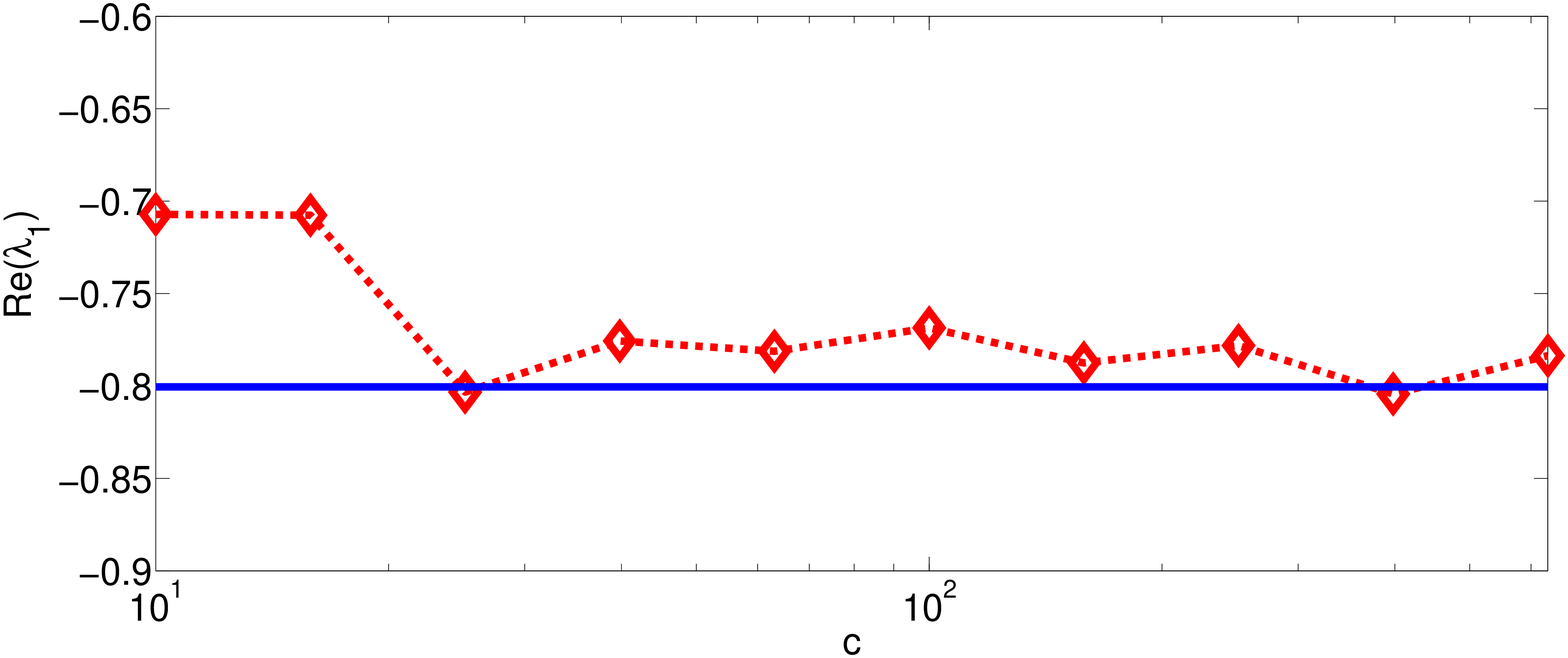}
}
\caption{Variation of ${\Re}(\lambda_{1})$ with large values of $c$, calculated for $f=0.3$, $\tau_{r}=\tau_{p}=0.1$, and mean degree $=3.4$. The estimation ${\Re}(\lambda_{1,\mathrm{est}})$ is plotted by the blue solid line and the real values  by the dash line with red $\diamond$.}\label{Fig5}
\end{figure}


We have shown in this paper that the stability of the multi-agent systems with a local pinning strategy and transmission delay may be destroyed by sufficiently large pinning delays.  Using theoretical and numerical methods, we have obtained an upper-bound for the delay value such that the system is stable for any pinning delay less than this bound. In this case, the exponential convergence rate of the multi-agent, which equals the smallest nonzero real part of the eigenvalues of the characteristic equation, measures the control performance.

\end{document}